\documentclass[12pt, a4paper]{article}

\usepackage{microtype}
\usepackage{graphicx} % pictures
\usepackage[ruled,vlined,linesnumbered]{algorithm2e}
\providecommand{\DontPrintSemicolon}{\dontprintsemicolon}
\usepackage{amsmath,amsfonts,amsthm,amssymb}
\usepackage[retainorgcmds]{IEEEtrantools}
\usepackage{fourier}

\theoremstyle{plain}
\newtheorem{theorem}{Theorem}[section]

\newtheorem{lemma}[theorem]{Lemma}
\newtheorem{proposition}[theorem]{Proposition}

\newtheorem{myclaim}[theorem]{Claim}
\newtheorem{definition}[theorem]{Definition}
\newtheorem{corollary}[theorem]{Corollary}

\theoremstyle{remark}

\newtheorem{remark}{Remark}

\newenvironment{proofof}[1]
{\noindent{\em Proof of Claim}}{\hfill {\tiny \qed}\smallskip}

\renewcommand{\vec}[1]{\mathbf{#1}}
\DeclareMathOperator*{\argmax}{arg\,max}

\newcommand{\MAX}{\mathrm{MAX}}
\newcommand{\ALG}{\mathrm{ALG}}
\newcommand{\OPT}{\mathrm{OPT}}
\newcommand{\SP}{\mathrm{SP}}
\newcommand{\ND}{\mathrm{ND}}

\allowdisplaybreaks

\DeclareMathOperator{\val}{val}

\title{Inequity Aversion Pricing over Social Networks: Approximation Algorithms and Hardness Results\footnote{A preliminary conference version appeared in MFCS 2016 \cite{AMS16}.}}

\author{
Georgios Amanatidis%\thanks{Athens University of Economics and Business, Athens, Greece.}
\hspace{-10pt} 
\and
Peter Fulla%\thanks{University of Oxford, Oxford, United Kingdom.}
\hspace{-10pt} 
\and
Evangelos Markakis%\footnotemark[2]
\hspace{-10pt}
\and
Krzysztof Sornat%\thanks{University of Wroclaw, Wroclaw, Poland.}
}

\begin{document}
\maketitle

\begin{abstract}
	We study a revenue maximization problem in the context of social networks. Namely, we consider a model introduced by \cite{AMT13} 
	that captures {\it inequity aversion}, i.e., prices offered to neighboring vertices should not be significantly different. We first provide approximation algorithms for a natural class of instances, referred to as the class of single-value revenue functions. Our results improve on the current state of the art, especially when the number of distinct prices is small. This applies, for example, to settings where the seller will only consider a fixed number of discount types or special offers. We then resolve one of the open questions posed in \cite{AMT13}, by establishing APX-hardness for the problem.  
	Surprisingly, we further show that the problem is NP-complete even when the price differences are allowed to be large, or even when the number of allowed distinct prices is as small as three. Finally, we provide some extensions of the model, regarding either the allowed set of prices, or the demand type of the clients.   
%	\keywords{Approximation Algorithms \and Revenue Maximization \and Inequity Aversion \and Social Networks \and Pricing}
\end{abstract}

%==================================================
%==================================================
%==================================================
\section{Introduction}
%==================================================
%==================================================
%==================================================
%P1 - generally about differential pricing
We study a differential pricing optimization problem in the presence of network effects. Differential pricing is a well known practice in everyday life and refers to offering a different price to potential customers for the same service or good. Examples include offering cheaper prices when launching a new product, making special offers to gold and silver members of an airline miles program, offering discounts at stores during selected periods, and several others.

%P2 over social networks => externalities
We are interested in studying differential pricing in the context of a social network. Imagine a network connecting individuals (who are seen as potential clients here) with their friends, family, colleagues, or other people who can exert some influence on them. One can have in mind other forms of abstract networks as well, e.g., a node could represent a geographic region, a neighborhood within a city, a type of profession, a social class, and edges can represent interactions or proximity. The presence of such a network creates {\it externality effects}, meaning that the decision of a node to acquire a new product or a new service is affected by the fact that some other nodes within her social circle (her neighborhood in the graph) already did so.
A typical example of positive externalities is when someone becomes more likely to buy a new product due to the positive reviews by a friend who already bought it in the past. Modeling positive externalities has led to a series of works that study marketing strategies for maximizing the diffusion of a new product, \cite{DR01,KKT03}, or the total revenue achieved, \cite{HMS08} (see also the Related Work section).      

%P3 - Negative externalities and inequity aversion
However, there also exist negative externality effects that can arise in a network. One example is the purchase of a product with the intention to show off and be a locally unique owner, e.g., a new type of expensive car, or clothes (also referred to as {\it invidious consumption}, see \cite{CCHW15}). In such a case, a node may be deterred from buying the same product, if a neighboring node already did so. A second example of negative externalities, which is the focus of our work, and arises from differential pricing, is {\it inequity aversion}, see e.g., \cite{BO00} and \cite{FS99}. % (Section 2). 
This simply means that a customer may experience dissatisfaction if she realizes that other people within her social circle, were offered a better deal for the same service. Hence, significant price differences, can create a negative response of some customers towards a product. 
Inequity aversion can also arise under a different, but equally applicable, interpretation: nodes may correspond to retail stores and an edge can signify proximity, so that clients could choose among these stores. Again, having significantly different prices to the same products is not desirable. 
%may lead to one of the stores losing its clientelle or getting a bad reputation. 

%P4 - The model of AMT
To capture the need for avoiding such phenomena, the relatively recent work of \cite{AMT13} introduced a model for pricing nodes over a social network. The main idea is to impose  constraints on each edge, specifying that the price difference between two neighbors should be bounded by some (endogenous) parameter, determined by the two neighbors.
On top of this, the seller is also allowed to not make a price offer to some nodes (referred to as introducing \emph{discontinuities}, see the related discussion in Section \ref{sec:definitions}), in which case the difference constraints do not apply for the edges incident to these nodes. 
Another way to interpret the model of discontinuities can be as follows: instead of discontinuities, we could allow price offers that would violate some edge constraints. In that case, if there is a node who has been offered a very high price, and there is such a violation, he would feel envious of some neighbor, and would choose not to buy the product. Hence, the seller would essentially not hope to extract any revenue from such nodes.

Assuming a finite set of available prices, unit-demand users, and digital goods (i.e., the supply can cover all the demand) the problem is to find a feasible price vector that satisfies the edge constraints and maximizes the total revenue. In its more general form the problem was shown to be NP-complete, and exact or approximation algorithms were derived for some interesting cases in \cite{AMT13}.
However, several questions remained open regarding the approximability status of the problem. 
\medskip   

\noindent {\bf Contribution:}
We revisit the model introduced by \cite{AMT13} (namely Model II in their work, which is the more general one), and study the approximability of the underlying revenue maximization problem. We resolve one of the open questions posed in \cite{AMT13}, regarding the complexity of the problem under the natural class of the so-called {\it single-value} revenue functions. Simply put, this means that the revenue extracted by each node is exactly the price offered to her, as long as the price does not exceed her valuation for the product (the usual assumption made in auction settings as well). We first establish APX-hardness for this class answering one of the open questions of \cite{AMT13}, and we also show that the problem is NP-complete even when the price differences are allowed to be relatively large (a case that could be thought easier to handle). Furthermore, we also show NP-hardness when we have only three distinct prices allowed, in contrast to the case of two distinct price offers, which is polynomial time solvable. 
%In doing so, we establish APX-hardness for the optimization problem of $3$-Terminal Node Cut, which can be interesting in its own right (so far APX-hardness was known for Multi-Terminal Node Cut, with a non-constant number of terminals \cite{GVY04}). 
We then provide approximation algorithms that improve some of the currently known results. Our improvement is stronger when the number of distinct prices is small. This applies for example to many settings where the seller will only consider a fixed number of discount types or special offers to selected customers. As the number of available price offers becomes large, the performance of our algorithm degrades to a logarithmic approximation. Finally, we also provide two extensions of these results; the first concerns a more general model where the allowed prices come from a set of $k$ arbitrary integers, instead of using price sets of the form $\{1, 2, \ldots, k\}$, as done in \cite{AMT13}, and the second concerns a multi-unit demand setting (see Sections \ref{sec:general} and \ref{sec:slope-svrf}).  \medskip 

%\subsection{Related Work}
\noindent {\bf Related Work:} Price discrimination is well studied in various domains in economics and is also being applied to numerous real life scenarios. The algorithmic problem of differential pricing over social networks is a more recent topic, initiated by \cite{HMS08}. The work of \cite{HMS08} studied a model with positive externalities, where the valuation of a player may increase as more friends acquire a good, and analyzed the performance of a very intuitive class of pricing strategies. Further improvements on the performance of such strategies were obtained later on by \cite{FS12}. The work of \cite{AGHMMN10} also considers a pricing problem but in an iterative fashion, where the seller is allowed to reprice a good in future rounds. Revenue maximization under a mechanism design approach was also taken in \cite{HIMM11} under positive network externalities. Finally, positive externalities have been used to model the diffusion of products on a network, see, among others, %\cite{DR01,KKT03,KKT05} and 
the exposition in \cite{Kleinberg07}.

Negative externalities within networks, as we focus on here, are less studied in the literature. For the concept of inequity aversion, see e.g., \cite{BO00,FS99}. The work most closely related to ours is \cite{AMT13}, which introduced the model considered here. Efficient algorithms were obtained for the case where discontinuities are not allowed (even for more general revenue functions), and also for networks with bounded treewidth. An approximation ratio of $1/(\Delta+1)$ was also provided, where $\Delta$ is the maximum degree. Similar results were shown for a stochastic version of the model. 
%NP-hardness was also established but for a more general class of functions than the ones we consider here.
Finally, other types of negative externalities have been considered e.g., in \cite{BKMX11,CCHW15} which study the effects of invidious consumption.

%===========================================
%===========================================
%===========================================
\section{Definitions and Preliminaries}\label{sec:definitions}
%==================================================
%==================================================
%==================================================

As usual, the social network in the model we study is represented as an undirected graph $G = (V, E)$, with $|V| = n$. 
The nodes depict the potential customers, and we consider a provider of some good or service, who has a finite set $P$ of available prices that he could offer to the nodes. In most of our presentation, we assume, as in \cite{AMT13}, that the available prices are given by $P = \{1, 2, \ldots, k\}$.
In Section \ref{sec:general}, we show how to extend the analysis when $P$ is an arbitrary set of $k$ positive integers, i.e., $P = \{p_1, p_2, \ldots, p_k\}$.

We further assume that every node has a unit-demand for the same product and that the supply of the seller is enough to cover the demand of all nodes.
For every node $v\in V$, we associate a revenue function $R_v:\{1,2,\ldots,k\} \mapsto \mathbb{N}$, 
that maps an offered price $p_v$ to the revenue that the provider gains from this offer. In this 
work, we focus on a simple and intuitive class of revenue functions, also studied in \cite{AMT13}. 
In particular, for a node $v\in V$, $R_v$ is called a \textit{single value revenue function}, if 
there exists a value $\val(v)$ such that when offered a price $p_v$:
\[R_v(p_v) =
\begin{cases}
p_v & \text{if } \val(v) \geqslant p_v \\
0 & \text{if } \val(v) < p_v
\end{cases}\]

We assume from now on that every node has a single value revenue function.
For the problem we study, we can also assume, without loss of generality, that $\val(v)\in P$, for every $v\in V$. To see this, note that for revenue maximization, that we are interested in, nodes with $\val(v)>k$ can only yield a revenue of $k$ and could be replaced by $\val(v)=k$, i.e., the highest possible price. Also for values that are less than $k$ but not integers, we can again extract only an integer revenue, given the form of $P$, so we could round them to the next integer that is smaller than $\val(v)$. Finally, any node $v$ with $\val(v)<1$ can be deleted without affecting the optimal revenue (see the concept of \textit{discontinuity} defined below), so we can completely ignore such nodes to begin with. Thus, we consider only instances with $\val(v)\in \{1, 2, \ldots, k\}, \forall v\in V$ (or with $\val(v) \in \{p_1, p_2, \ldots, p_k\}$, for the generalized results of Section \ref{sec:general}).

Given a vector $\vec{p} = (p_v)_{v\in V}$ of prices offered to the nodes, %with $\vec{p}\in {P}^n$, 
we use $R(\vec{p})$ to denote 
the total revenue, i.e., 
$R(\vec{p}) = \allowbreak \sum_{v \in V} R_v(p_v)$. Our goal is to find a price vector that maximizes the total revenue. At the same time, however, we want to capture the effect of {\it inequity aversion} \cite{BO00,FS99} in social networks. This means that a node may experience dissatisfaction if she sees that other nodes within her social circle, were offered a better deal for the same service. Hence, significant price differences, create  negative externalities among users. 

To avoid such situations the model introduced in \cite{AMT13} has constraints on each edge, stating that the  price difference between two neighbors $u,v$ is bounded, i.e., $p_u - p_v \leqslant \alpha(u, v)$ and $p_v - p_u \leqslant \alpha(v, u)$, for every $(u, v)\in E$. Here, $\alpha(\cdot, \cdot) \geqslant 0$ is integer-valued (given that the prices are also integers), 
and it can be non-symmetric. Furthermore, the seller is also allowed not to make an offer to certain nodes. Formally, this is captured by having one more price option, which we denote by $\bot$, with $R_v(\perp)=0$. Setting $p_v = \bot$ to a node means that the provider does not make any offer to $v$
and there is no price restriction on the edges that are incident to $v$. We can essentially think about this as deleting these vertices from the graph. We will refer to setting $p_v = \bot$ to a node $v\in V$, as introducing a {\it discontinuity} on $v$. 
Avoiding making an offer can be thought of as choosing not to promote a product or service within a certain region or within a certain social group.  
In terms of optimization, discontinuities can help the seller in producing much higher revenue, than without discontinuities, as Proposition \ref{prop:h_n} in Section \ref{sec:sp} states.

There is another way to interpret the notion of discontinuities: Suppose we had no discontinuities, but at the same time we were allowed to make offers that could violate some price constraints along edges. In that case, suppose there is a node who has been offered a high price; this  price may or may not be higher than the maximum price the node is willing to pay. In the presence of an edge violation, however, he would feel envious of some neighbor, and being dissatisfied he would choose not to buy the product. Hence, the seller does not hope to extract any revenue from such nodes. This is equivalent to our model of just setting a discontinuity and only caring for the remaining edge constraints, not involving this node. For convenience, we will stick to allowing discontinuities, rather than making price offers that could violate the constraints.

Given this model, the set of feasible price vectors is then:
${\cal F} = \{\vec{p} : \forall ~v\in V, p_v\in P\cup \{\bot\},\allowbreak \text{\ and\ } \forall ~(u, v)\in E,\ p_u\neq\bot ~\wedge~ p_v\neq\bot \Rightarrow p_u - p_v \leqslant \alpha(u, v) ~\wedge~ p_v - p_u \leqslant \alpha(v, u) \}$.
Therefore, the problem we study is:\smallskip

\noindent \textbf{Inequity Aversion Pricing:}
Given a graph with edge constraints, and a single-value revenue function for each node, find a feasible price vector that maximizes the total revenue, i.e., find $\vec{q}\in \argmax_{\vec{p}\in{\cal F} } \;\sum_{v \in V} R_v(p_v)$.\smallskip

Some cases of this problem, as well as the variant where no discontinuities are allowed, are already known to be polynomial time solvable \cite{AMT13}. 
Regarding hardness, although the problem is NP-hard for more general revenue functions, it was posed as an open question whether NP-hardness still holds for single value revenue functions (the hardness result in \cite{AMT13} requires instances with revenue functions that cannot be captured by single value ones).

%==================================================
%==================================================
%==================================================
\section{Warmup: Basic Facts and Single-price Solutions}
\label{sec:sp}
%==================================================
%==================================================
%==================================================

In this section, we present a simple algorithm and some basic observations, which we use later on, in Section \ref{sec:k-values}. 

Let $v_{\max} = \max_{v \in V} \val(v)\leqslant k$, and $\MAX = \sum_{v \in V} \val(v)$. Given an instance of the problem, we denote by $\OPT$ the revenue of an optimal solution. The quantity $\MAX$ is clearly an upper bound on the optimal revenue, hence $\OPT \leqslant \MAX$.

We will refer to a solution as being a {\it single-price} solution, if it charges the same price to every node without introducing discontinuities. This is always a feasible solution since all the edge constraints are satisfied. The revenue extracted by a single-price algorithm that uses the price of $p$ for all nodes is equal to $\ p\cdot |\{v\in V: \val(v)\geqslant p\}|$.

To understand whether a single-price solution can be of any help for our setting, we can examine the performance of the best possible single price. The following observation suggests that we do not need to try too many values, even if $v_{max}$ is very large.

\begin{lemma}\label{number_of_OP}
	In order to find the optimal single-price solution, it suffices to check at most $\min\big\{n,\allowbreak v_{\max}\big\}$ possible prices.
\end{lemma}
\begin{proof}
	There are at most $min\{n, v_{\max}\}$ different values in the set $\{\val(v): v \in V\}$. It is never optimal to use any price $p\notin \{\val(v): v \in V\}$. Indeed, if $p \in (\val(v_1),\val(v_2))$, where $\val(v_1)$ and $\val(v_2)$ are two consecutive distinct values for some nodes $v_1, v_2\in V$, then it is strictly better to set the price to $\val(v_2)$. For the same reason, it is suboptimal to set a price that is less than the minimum value across nodes, while if we use a price $p > v_{\max}$ then we gain no revenue. %\qed
\end{proof}

Hence in  $O(\min\{n, v_{\max}\})$ steps, we can select the best single-price solution. Let us denote by $R_{\SP}$ the revenue raised by this solution. The performance of $R_{\SP}$ has been analyzed in a different context\footnote{The work of \cite{GHKSW06} studied an auction pricing problem without the presence of social networks.} by \cite{GHKSW06}, where it was shown that it achieves a $\Theta(\ln{n})$-approximation. Here we give a slightly tighter statement, which we utilize in later sections for small values of $v_{max}$. 

\begin{theorem}\label{thm:h_n}
	For any number $n$ of agents, the optimal single-price solution achieves a $1/H_r$-approximation, where $r=\min\{n, v_{\max}\}$, and $H_{\ell}$ is the $\ell$-th harmonic number, i.e.,
	\[ R_{\SP} \geqslant \frac{\MAX}{H_r} \geqslant \frac{\OPT}{H_r} \,. \] 
	Furthermore, the approximation guarantee is tight.
\end{theorem}

The proof follows from the proof of Theorem \ref{thm:h_n_gen} in Section \ref{sec:general}, which provides a more general result.
One interesting point here is that single-price solutions do not use any discontinuities. 
If $R_{\ND}$ is the optimal revenue without using any discontinuities, clearly $R_{\ND}\geqslant R_{\SP}$.
And as we mentioned in Section \ref{sec:definitions}, it is possible to find the optimal solution 
that does not use discontinuities in polynomial time \cite{AMT13}; so why use something worse instead of using directly $R_{\ND}$? Actually, for our purposes, besides being harder to argue about, $R_{\ND}$ turns out to be as bad an approximation as $R_{\SP}$ in the worst case, when allowing discontinuities.

The proposition below provides some further justification for the model with discontinuities, under the seller's point of view. In particular, it reveals that introducing discontinuities can cause a significant increase in the optimal revenue achievable by the seller, compared to what can be achieved without discontinuities.

\begin{proposition}\label{prop:h_n}
	The optimal solution with no discontinuities achieves a $1/H_r$-approximation, where $r=\min\{n, v_{\max}\}$, and this approximation guarantee is tight.
\end{proposition}

The approximation guarantee follows from the fact that single-price solutions do not use any discontinuities. To see that without using any discontinuities one cannot always do better, we can modify slightly the examples that give tightness within the proof of Theorem \ref{thm:h_n_gen} (see Remark \ref{rem:prop1} after the proof of Theorem \ref{thm:h_n_gen}).

%==================================================
%==================================================
%==================================================
\section{Approximation Algorithms for Inequity Aversion Pricing}
\label{sec:k-values}
%==================================================
%==================================================
%==================================================
In this section we present new approximation algorithms for the problem by exploiting ways in which setting discontinuities in certain nodes can help. Our main result is an approximation algorithm, with a ratio of $(H_k-0.5)^{-1}$. Even though asymptotically this is no better than the optimal single-price algorithm, it does yield better ratios for instances where $k$ is a small constant. The motivation for studying cases where the set of available prices is small is that a seller may be willing to offer only specific types of discount to selected customers, e.g., $20\%$ or $30\%$ off the regular price and so on, rather than using an arbitrary set of prices.  

We will begin with an exact algorithm for the case of $P = \{1, 2\}$, before we move to having $P = \{1, 2, \ldots, k\}$. As we will see in Section \ref{sec:hard}, we can hope for an exact algorithm only for $k=2$, since the problem becomes hard for higher values of $k$. 
In fact, since our proposed algorithm also works for an arbitrary set of two integer prices, we will directly present the case of $P = \{p_1, p_2\}$, for two positive integers, $p_1, p_2$.

%%%%%%%%%%%%%%%%%%%%%%%%%%%%%%%%%%%%%%%%%% 
\subsection{An Exact Algorithm when $P = \{p_1, p_2\}$}
\label{sec:IS}
%%%%%%%%%%%%%%%%%%%%%%%%%%%%%%%%%%%%%%%%%%

In this subsection, we assume the available prices are $p_1$, $p_2$, and $\bot$. Without loss of generality, we assume that $p_1 < p_2$.
Given the discussion in Section \ref{sec:definitions}, we will also assume that for every node $v\in V$, $\val(v)\in\{p_1, p_2\}$.
Even when $p_1=1$ and $p_2=2$, the problem still remains non-trivial. 
For such instances we already have a $\frac{2}{3}$-approximation by Theorem~\ref{thm:h_n}, that does not use discontinuities. The difficulty in improving this factor is in finding a way of selecting appropriate nodes to set to $\bot$.

Before we formally state our algorithm, let us  illustrate the main idea.
Consider an instance of the problem on a graph $G = (V, E)$. 
We construct an appropriate bipartite graph $H$ such that feasible price vectors 
for $G$ correspond to independent sets of $H$. Hence, the problem reduces to finding a maximum 
weighted independent set in bipartite graphs, which is well known to be solvable in polynomial time.
To be more specific, denote by $A, B \subseteq V$ 
the partition of the vertices of $G$ into vertices with $\val(v) = p_1$ 
and $\val(v) = p_2$ respectively.
To construct the bipartite graph $H = (V_1 \cup V_2, E')$, we will be using a superscript for every node, to clarify whether it belongs to $V_1$ or $V_2$. For each $a \in A$, the graph $H$ will have 
a vertex $a^1 \in V_1$; for each $b\in B$, $H$ will have two vertices 
$b^1 \in V_1$, and $b^2 \in V_2$ connected with an edge.
Additionally, we include an edge between $x^2$ and $y^1$ for every $x \in B, y
\in V$ for which $(x, y) \in E(G)$ and $\alpha(x, y) < p_2 - p_1$. Note that $H$ is
bipartite, since there is no edge between vertices of the same superscript.

Now consider any independent set $S$ of $H$. We interpret $u^1
\in S$ as offering price $p_1$ to vertex $u$ in the original instance and $u^2 \in S$ as offering price $p_2$; note that it cannot be the case that both $u^1$ and $u^2$ belong to $S$. 
If none of $u^1, u^2$ belong to $S$, we  interpret this as introducing a discontinuity on $u$. 
To see that this is a feasible price vector, we only have to worry about edges $(x, y)$ in the graph $G$, for which $\alpha(x, y) < p_2 - p_1$, since all other constraints are trivially satisfied.  
But by construction, the grah $H$ has an edge $(x^2, y^1)$ whenever $\alpha(x, y) < p_2 - p_1$. Hence, we cannot include both $x^2$ and $y^1$ in $S$, and this ensures that either we offer the same price to such nodes, or one of them will have a discontinuity, implying that the resulting price vector is feasible. 
Conversely, it is easy to see that any feasible price vector corresponds to an independent set of $H$; given such a vector for the vertices of $G$, if price $p_i \in \{p_1, p_2\}$ is offered to $v$, we include the vertex $v^i$ in the independent set of $H$.

We can now make $H$ weighted by setting weight $p_i\in\{p_1, p_2\}$ to each vertex $v^i$, for $i=1, 2$. Then, the total
weight of an independent set equals the total revenue of the corresponding price
vector and vice versa. Thus, in order to solve the Inequity Aversion Pricing on the original instance, it suffices to find a maximum weight independent set of $H$.

\begin{algorithm}[h]
	\DontPrintSemicolon 
	{
		Given the graph $G = (V, E)$, construct the bipartite graph $H$ with  $V(H) = \allowbreak \{v^1 \ | \ v\in V(G)\}\cup \{v^2 \ | \ v\in V(G) \text{ and } \val(v) = p_2\}$ and $E(H) = E_1 \cup E_2$, where  $E_1 = \{(v^1, v^2) \ | \ v\in V(G) \text{ and } \val(v) = p_2\}$ and $E_2 = \{(u^1, v^2) \ | \ (u, v)\in E(G), \allowbreak \val(v) = p_2  \text{ and }  \alpha(v, u) < p_2-p_1\}$\;
		Find a maximum weight independent set on $H$, say $S\subseteq V(H)$\;
		For every $u^i \in S$, offer a price of $p_i$ to the corresponding vertex $u$ of $V(G)$ \;
		Set $\bot$ to all the remaining vertices of $V(G)$\;
		Return the resulting price vector \;
	}
	\caption{An exact algorithm when $P = \{p_1, p_2\}$}\label{fig:alg-is}
\end{algorithm} %\vspace{-10pt}

The next theorem summarizes the above discussion.

\begin{theorem}\label{thm:issolution}
	Algorithm \ref{fig:alg-is} solves optimally the Inequity Aversion Pricing problem when $P = \{p_1, p_2\}$ in polynomial time. 
\end{theorem}

%==================================================
\subsection{An Approximation Algorithm for $k>2$}
%==================================================

We now consider the case where  there are more than two available prices. In order to improve the approximation guarantee of Theorem \ref{thm:h_n}, we reduce the problem to the case of $k=2$, and use the results of the previous subsection. 

Consider an instance of the problem, with available prices in $\{\perp, 1, 2, \ldots, k\}$. As discussed in Section \ref{sec:definitions}, we may assume that $\val(v)\in \{1, 2, \ldots, k\}$ for every $v\in V$. We create now another instance, where we set the value of every node with $\val(v)>1$ to be equal to $2$. We can then run Algorithm \ref{fig:alg-is} from Subsection \ref{sec:IS} on this new instance. At the same time, we can also compute the optimal single-price solution for the original instance, and pick the best among these two solutions. This yields Algorithm \ref{fig:alg-k>2}, described below.

\begin{algorithm}[h]
	\DontPrintSemicolon 
	{\small
		Given an instance $I$, construct a new instance $I'$, where for every $v\in V$,  $\val'(v)  = \min\{\val(v), 2\}$; everything else remains unchanged \;
		Run Algorithm \ref{fig:alg-is} from Subsection \ref{sec:IS} on instance $I'$, and let $R_*$ be the revenue obtained\;
		Compute the optimal single-price solution without discontinuities, on the original instance $I$, as described in Section \ref{sec:sp}, with revenue $R_{\SP}$ \;
		Return the solution that achieves $\max \{R_*, R_{\SP}\}$\;
	}
	\caption{An algorithm for $k>2$}\label{fig:alg-k>2}
\end{algorithm}

Clearly, Algorithm \ref{fig:alg-k>2} runs in polynomial time. 
Note that the solution returned by the algorithm is feasible. Any single-price solution is always feasible, while Algorithm \ref{fig:alg-is}
will produce a price vector that is feasible for $I'$, and therefore for $I$ since the edge restrictions in the two instances are the same.
Even though asymptotically, this is still a logarithmic approximation, the algorithm achieves significantly better results for small values of $k$.

\begin{theorem}\label{thm:alg-k>2}
	Algorithm \ref{fig:alg-k>2} achieves a $\frac{1}{H_{v_{\max}} - 0.5}$-approximation ratio for Inequity Aversion Pricing when the available prices are $\{\perp, 1,2, \cdots,k\}$, with $k\geqslant 2$.
\end{theorem}

\begin{proof}
	The proof is by induction on $v_{\max}$. For $v_{\max}=2$ the result follows from Theorem~\ref{thm:issolution} since $1 = \frac{1}{H_2 - 0.5}$.
	
	Now assume we have an instance $I$ where $v_{\max}= j > 2$.  As usual, let $\OPT$ denote the optimal revenue for $I$ and $\ALG$ the revenue returned by Algorithm \ref{fig:alg-k>2}. Also, let $R_{j}$ be the revenue extracted by setting price $j$ at every node, and $V_j = \{v \in V: \val(v)=j\}$. We consider two cases. %\smallskip
	
	\noindent \emph{Case (i):}  $|V_j| \geqslant \frac{1}{\left( H_j - 0.5\right)  j} \cdot \OPT$.
	Then, $\frac{\ALG}{\OPT} \geqslant \frac{R_j}{\OPT} = \frac{j \cdot |V_j|}{\OPT} \geqslant \frac{\frac{1}{H_j - 0.5} \cdot \OPT}{\OPT} = \frac{1}{H_j - 0.5}$. \medskip

	\noindent \emph{Case (ii):} $|V_j| < \frac{1}{\left( H_j - 0.5\right)  j} \cdot \OPT$.
	Let $I^*$ be an instance derived from $I$ by setting $\val^*(v) = \min\{\val(v),\allowbreak j-1\}$, i.e., we only reduce the valuation of the nodes with $\val(v)= v_{\max}$ by $1$.
	Let $\OPT^*$ and $\ALG^*$   denote the optimal revenue and the revenue returned by Algorithm \ref{fig:alg-k>2} respectively, given  $I^*$.
	By the inductive hypothesis we have $\ALG^* \geqslant \frac{1}{H_{j-1} - 0.5} \cdot \OPT^*$.
	
	Furthermore, notice that the set of vertices with valuation greater than $1$ is the same in both instances. So, Algorithm \ref{fig:alg-k>2} on input $I^*$ considers exactly the same price vectors as it does on input $I$, with the exception of the single-price solution that universally uses $j$. We conclude that $\ALG^* \leqslant \ALG$. Next, we prove the following useful claim.

	\begin{myclaim}\label{claim:opt'}
		$\OPT^* \geqslant \OPT -|V_j|$.
	\end{myclaim}
	\begin{proofof}{Proof of Claim \ref{claim:opt'}} %\renewcommand{\qed}{\hfill $\triangleleft$}
		Let $\mathbf p$ be an optimal price vector for $I$. Construct the price vector $\mathbf p^*$ by decreasing any price that is at least $j$ to $j-1$.
		It is straightforward to see that in instance $I$ we have $R(\mathbf p^*) \geqslant R(\mathbf p)-|V_j| = \OPT -|V_j|$, while in both instances $R(\mathbf p^*)$ is the same. What is left to show is that $\mathbf p^*$ is feasible  for $I^*$. Observe, however, that the two instances have exactly the same edge restrictions and that $\mathbf p^*$ did not increase the price difference between any two vertices compared to  $\mathbf p$.
		Thus, $\OPT^* \geqslant R(\mathbf p^*) \geqslant \OPT -|V_j|$. %\qed
	\end{proofof}

	\noindent Now, we can write
	\begin{IEEEeqnarray*}{rCl}
		\frac{\ALG}{\OPT} & \geqslant & \frac{\ALG^*}{\OPT} \geqslant \frac{\frac{1}{H_{j-1} - 0.5} \cdot \OPT^*}{\OPT} \geqslant \frac{\frac{1}{H_{j-1} - 0.5} \cdot (\OPT-|V_j|)}{\OPT}\\
		& \geqslant & \frac{1}{H_{j-1} - 0.5} \left(  1 -\frac{ \frac{1}{j(H_{j}  - 0.5)}\cdot \OPT}{\OPT} \right) = \frac{1}{H_{j-1} - 0.5} \cdot   \frac{j H_j - 0.5 j - 1}{j(H_{j}  - 0.5)}   \\
		& = &  \frac{1}{H_{j-1} - 0.5} \cdot  \frac{j (H_{j-1} - 0.5)}{j(H_{j}  - 0.5)} = \frac{1}{H_{j} - 0.5}  \,,
	\end{IEEEeqnarray*}
	which concludes the proof. %\qed
\end{proof}

%==================================================
%==================================================
%==================================================
\section{Approximation Algorithms for General Price Sets}\label{sec:general}
%==================================================
%==================================================
%==================================================

In this section we extend our results when $P$ is an arbitrary set of $k$ positive integers, i.e., $P = \{p_1, p_2, \ldots, p_k\}$. This can be seen as a more realistic model, especially for small values of $k$. In such a case, one could try to directly apply Theorems \ref{thm:h_n}, \ref{thm:issolution}, or \ref{thm:alg-k>2} for $P'=\{1, 2, 3, 4, \ldots , p_k\}$. However, this may produce a very poor approximation when $k$ is small but $p_k$ is large, and feasibility is not guaranteed either. 

In what follows,  $P_j$  denotes $\sum_{i=1}^{j}\frac{p_i - p_{i-1}}{p_i}$, where $p_0=0$.
We begin with a generalization of Theorem \ref{thm:h_n}. 

\begin{theorem}\label{thm:h_n_gen}
	For any number $n$ of agents and possible prices $p_1< p_2 < \ldots < p_k$, the optimal single-price algorithm achieves a $\rho$-approximation, where $\rho=1 / \min \{H_n, P_k \}$, i.e.,
	\[ R_{\SP} \geqslant \frac{\MAX}{\min \{H_n, P_k\}} \geqslant \frac{\OPT}{\min \{H_n, P_k\}} \,, \] 
	and this approximation guarantee is tight.
\end{theorem}
\begin{proof} %[Proof of Theorem \ref{thm:h_n_gen}]
	As in Lemma \ref{number_of_OP}, we know that we only need to consider at most $min\{n, k\}$
	possible prices that correspond to the distinct values of the nodes. 
	Let $a_j$ be the number of vertices with value $p_j$ and $R_i$ be the revenue obtained by setting the price of all nodes equal to $p_i$, i.e.,
	\[a_j = \big|\{ v \in V: \val(v)=p_j\}\big| \text{\ \ and\ \ }R_i = \sum_{v \in V} R_v(i) =  p_i \cdot \sum_{j=i}^n a_j \,.\]

	Recall that $P_k = \sum_{i=1}^{k} \frac{p_i - p_{i-1}}{p_i}$, where $p_0=0$. Let $R_{\SP}$ be the revenue achieved by the optimal one-price algorithm. Then $R_i\leqslant R_{\SP}$, and 
	we have
	\[ \MAX = \sum_{v \in V} \val(v) = \sum_{i=1}^{k} \bigg((p_i - p_{i-1})\cdot\sum_{j=i}^n a_j\bigg) = \sum_{i=1}^{k} \frac{(p_i - p_{i-1})\cdot R_i}{p_i} \leqslant R_{\SP} \cdot P_k \,.\]
	So, we obtain
	\begin{equation}\label{op_lnv}
	R_{\SP} \geqslant \frac{\MAX}{P_k} \geqslant \frac{\OPT}{P_k}\,. 
	\end{equation}
	
	\noindent
	Let us now sort the vertices from $V$ with respect to $\val(v)$ in ascending order, say $v_1,\ldots, v_n$. Let $R_{(i)}$ be the revenue obtained from the vertices $\{v_i,v_{i+1},\dots, \allowbreak v_n\}$ by setting to all of them the price $\val(v_i)$, i.e.,
	\[ R_{(i)} = \mkern-10mu \sum_{v \in \{v_i,v_{i+1},\dots,v_n\}} \mkern-14mu R_v(\val(v_i)) = (n-i+1) \cdot \val(v_i)\,.\]
	Clearly, $R_{(i)} \leqslant R_{\SP}$ and we have 
	\[ \MAX = \sum_{i=1}^n \val(v_i) = \sum_{i=1}^{n} \frac{R_{(i)}}{n-i+1} \leqslant R_{\SP} \cdot \sum_{i=1}^{n} \frac{1}{n-i+1} %= R_{\SP} \cdot \sum_{i=1}^{n} \frac{1}{i} 
	= R_{\SP} \cdot H_n\,.\]
	Hence, we obtain
	\begin{equation}\label{op_lnn}
	R_{\SP} \geqslant \frac{\MAX}{H_n} \geqslant \frac{\OPT}{H_n}\,. 
	\end{equation}
	Putting inequalities \eqref{op_lnv} and \eqref{op_lnn} together  completes the proof.\medskip
	
	To see that this is tight, consider the following family of graphs.
	For any $n$ take $G(n)$ to be a clique on $\{v_1,v_2,\dots,v_n\}$ and let $\val(v_i)=\frac{n!}{i}$ and $\alpha(u,v)=k = n!$ for every edge. Then
	\[ \OPT = \MAX = \sum_{i=1}^n \frac{n!}{i} = n! H_n \text{\ \ \ and\ \ \ }  R_{\frac{n!}{i}} = \sum_{j=1}^i \frac{n!}{i} = n!\,, \ \   \forall {i \in \{1,2,\dots,n\}}.\]
	Therefore
	\[ \frac{R_{\SP}}{\OPT} = \frac{\max_{i \in \{1,2,\dots,n\}} R_{\frac{n!}{i}}}{n! H_n}  = \frac{1}{H_n}\,.\]
	
	In fact, tightness holds even when $P_k \leqslant H_n$.
	Consider an instance where $p_i=i, \forall i\in k$ and $n=k!$. 
	Define $G(k)$ to be a clique on $\bigcup_{i=1}^{k} V_i$, where $V_i = \{v \in G(k): \val(v)=i \}$ and
	\[ \forall i \in \{1,2,\dots,k-1\}, \quad |V_i| = \frac{k!}{i(i+1)}, \text{\ \ and\ \ } |V_{k}| = \frac{k!}{k} \,.\]
	Like before, $\alpha(u,v)=k$ for every edge. It is easy to verify that $\sum_{i=1}^{k} |V_i| = n$. Then
	\[ \OPT = \MAX = \sum_{i=1}^{k} i \cdot |V_i| = k \cdot \frac{n}{k} + \sum_{i=1}^{k-1} i \cdot \frac{n}{i(i+1)} = n \cdot H_{k} = n \cdot P_k \,, \]
	while
	\begin{IEEEeqnarray*}{rCl}
		\forall {i \in \{1,2,\dots,k\}}, \quad R_{i} & = & \sum_{j=i}^{k} i \cdot |V_j| = i \cdot n \cdot \Bigg( \frac{1}{k} + \sum_{j=i}^{k-1} \frac{1}{j(j+1)} \Bigg)\\
		& = & i \cdot n \cdot \Bigg( \frac{1}{k} + \sum_{j=i}^{k-1} \left( \frac{1}{j} - \frac{1}{j+1} \right) \Bigg) = i \cdot n \cdot \frac{1}{i} = n\,.
	\end{IEEEeqnarray*}
	Therefore,
	\[ \frac{R_{\SP}}{\OPT} = \frac{\max_{i \in \{1,2,\dots,k\}} R_{i}}{n \cdot P_k} = \frac{1}{P_k}\,. \] %\qed
\end{proof}	

\begin{remark}\label{rem:prop1}
	By slightly modifying the tightness examples from the proof of Theorem \ref{thm:h_n_gen} we can show that one cannot achieve a better guarantee than the one suggested by Proposition \ref{prop:h_n} without using discontinuities. 
	In each case, we connect a new vertex $v$ with value $1$ to every vertex $u$ and set $\alpha(u,v)=\alpha(v, u)=0$. The optimal solution is to put a  discontinuity on $v$ and maximize the revenue of every other vertex. When discontinuities are not allowed though, a solution cannot do better than $R_{SP}$, since  all prices have to be equal in such a solution. It is easy to see that we still get the same ratios, % of the proof of Theorem \ref{thm:h_n_gen}, 
	namely $\frac{1}{H_{n}}$ and $\frac{n+1}{n\cdot H_{v_{\max}}}$.
\end{remark}

We are now ready to state the generalization of Theorem~\ref{thm:alg-k>2} for arbitrary price sets.

\begin{theorem}\label{thm:alg-k>2_gen}
	Let $P= \{p_1, p_2, \cdots, p_k\}$. There exists a polynomial time $\frac{1}{P_k + \frac{p_1}{p_2}-1}$-approximation algorithm for the Inequity Aversion Pricing problem.
\end{theorem}

\begin{proof}
	The algorithm is very similar to Algorithm \ref{fig:alg-k>2}. The only difference is in the definition of the instance $I'$ (step 1 of Alg. \ref{fig:alg-k>2}). 
	Namely, given an instance $I$, let $I'$ be the new instance where for every $v\in V$,  $\val'(v) = \min\{\val(v), p_2\}$, for all $v\in V$, while the constraints remain the same.

	The proof is by induction on $k$. For $k=2$ the algorithm is equivalent to Algorithm~\ref{fig:alg-is} that gives an exact solution (see Theorem~\ref{thm:issolution}). We have 
	\[\frac{1}{P_k + \frac{p_1}{p_2}-1} = \frac{1}{\frac{p_1-p_0}{p_1} + \frac{p_2-p_1}{p_2} + \frac{p_1}{p_2}-1} = 1. \]
	
	Now assume we have an instance $I$ where $k> 2$. We use the notation of the proof of Theorem \ref{thm:alg-k>2}, but now $R_{k}$ is the revenue extracted by setting price $p_k$ at every node, and $V_k = \{v \in V: \val(v)=p_k\}$.
	
	\noindent \emph{Case (i):}  $|V_k| \geqslant \frac{1}{\left( P_k + \frac{p_1}{p_2}-1\right)  p_k} \cdot \OPT$.
	Then, $\frac{\ALG}{\OPT} \geqslant \frac{R_k}{\OPT} = \frac{p_k \cdot |V_k|}{\OPT} \geqslant  \frac{1}{P_k + \frac{p_1}{p_2}-1}$.

	\noindent \emph{Case (ii):} $|V_k| < \frac{1}{\left( P_k + \frac{p_1}{p_2}-1\right)  \cdot p_j} \cdot \OPT$.
	Let $I^*$ be an instance derived from $I$ by setting $\val^*(v) = \min\{\val(v), p_{k-1}\}$. 
	By the inductive hypothesis we have $\ALG^* \geqslant \frac{1}{P_{k-1} + \frac{p_1}{p_2}-1} \cdot \OPT^*$.
	It is easy to see that $\ALG^* \leqslant \ALG$. Also, we can prove an analog of Claim \ref{claim:opt'}, namely $\OPT^* \geqslant \OPT -(p_k-p_{k-1})|V_k|$.
	Putting everything together we have
	\begin{IEEEeqnarray*}{rCl}
		\frac{\ALG}{\OPT} & \geqslant & \frac{\ALG^*}{\OPT} \geqslant \frac{\frac{1}{P_{k-1} + \frac{p_1}{p_2}-1} \cdot \OPT^*}{\OPT} \geqslant \frac{\frac{1}{P_{k-1} + \frac{p_1}{p_2}-1} \cdot (\OPT-(p_k-p_{k-1})|V_k|)}{\OPT}\\
		& > & \frac{1}{P_{k-1} + \frac{p_1}{p_2}-1} \left(  1 -\frac{ \frac{p_k-p_{k-1}}{p_k(P_{k}  + \frac{p_1}{p_2}-1)}\cdot \OPT}{\OPT} \right)\\
		& = & \frac{1}{P_{k-1} + \frac{p_1}{p_2}-1} \cdot   \frac{\frac{p_k}{p_k-p_{k-1}} P_k + \frac{p_k}{p_k-p_{k-1}}(\frac{p_1}{p_2}-1)  - 1}{\frac{p_k}{p_k-p_{k-1}}(P_k  + \frac{p_1}{p_2}-1)}   \\
		& = &  \frac{1}{P_{k-1} + \frac{p_1}{p_2}-1} \cdot  \frac{\frac{p_k}{p_k-p_{k-1}} (P_{k-1} + \frac{p_1}{p_2}-1)}{\frac{p_k}{p_k-p_{k-1}}(P_k  + \frac{p_1}{p_2}-1)} = \frac{1}{P_k + \frac{p_1}{p_2}-1}  \,,
	\end{IEEEeqnarray*}
	which concludes the proof.
\end{proof}

Table 1 summarizes the approximation ratios obtained by three algorithms: the optimal single price solution, Algorithm \ref{fig:alg-k>2}, as well as the algorithm implied by Theorem \ref{thm:alg-k>2_gen} for different sets of prices. The positive news is that Theorem \ref{thm:alg-k>2_gen} is resistant to price scaling. For example, the price set $\{10,20,30\}$ can be seen as simply $\{1,2,3\}$ scaled up, and Theorem \ref{thm:alg-k>2_gen} yields the same ratio as Algorithm \ref{fig:alg-k>2} does for $P = \{1,2,3\}$. On the other hand, the $1/H_k$-factor algorithm gives a much worse approximation when $P = \{10,20,30\}$. 

\begin{table}[h]
	\centering
	\label{approx_ratios}
	\begin{tabular}{|c|c|c|c|c|c|}
		\hline 
		$P$ & \{1, 2\}  & \{1, 2, 3\} & {\footnotesize\{1, \ldots , 100\}} & \{10, 20, 30\} & \{70, 80, 90, 100\}  \\
		\hline
		$1/ H_k$  & 0.666 & 0.545 & 0.192 & 0.286 & 0.192  \\
		\hline
		Alg. \ref{fig:alg-k>2} & 1 & 0.750 & 0.213 & -- & -- \\
		\hline
		\footnotesize Thm. \ref{thm:alg-k>2_gen}  & 1 & 0.750 & 0.213 & 0.750 & 0.825 \\
		\hline
		
	\end{tabular}
	\vspace{0.2cm}
	\caption{Examples of obtained approximation ratios.}
\end{table} %\vspace{-15pt}

%==================================================
%==================================================
%==================================================
\section{Hardness for Single Value Revenue Functions \label{sec:hard}}
%==================================================
%==================================================
%==================================================
In \cite{AMT13} there is an $n^{1-\varepsilon}$ inapproximability result for Inequity Aversion Pricing, but for general revenue functions and $\alpha(u, v)=1$ for every edge.
An NP-hardness proof is also given for these edge constraints when single value and constant revenue functions are allowed. The NP-hardness of Inequity Aversion Pricing as we study it here, i.e., allowing only single value revenue functions, was left as an open question. We resolve this question by proving that the problem remains NP-complete even if we restrict the revenue functions to be single value. Our reduction implies that
the result holds even when the price differences are allowed to be quite large and
close to the maximum possible price $k$ (which could presumably make the problem easier). Further, when $\alpha(u, v)=0$ for every edge, we are able to show APX-hardness, as well as NP-hardness even for $k=3$, in contrast to the case of $k=2$.  

The reduction, below, is from the decision version of \emph{3-Terminal Node Cut}: Given a graph $G(V,E)$, 
a set $S= \{v_1, v_2, v_3\}\subseteq V$, and an integer $q$, is there a subset of $q$ vertices that 
can be deleted, so that any two vertices of $S$ are in different connected components of 
the resulting graph? 
The NP-completeness of the weighted version of 3-Terminal Node Cut is discussed in 
\cite{Cun89}, while the APX-hardness of the unweighted version we use here  is discussed in \cite{GVY04}.
%in either case no explicit proof is given. 
The NP-completeness result we need follows from Theorem \ref{MNC-hard} as well (see the discussion before the statement of Theorem \ref{MNC-hard}).

\begin{theorem}\label{np-comp}
	Let $\varepsilon > 0$ be any small constant. 
	The decision version of Inequity Aversion Pricing for single value revenue functions is NP-complete, even when $\alpha(u, v)$ is as large as $k^{1-\varepsilon}$  for all $(u, v)\in E(G)$, where $k$ is the maximum possible price.
\end{theorem}
\begin{proof}
	It is immediate that the problem is in NP.
	To facilitate the presentation, we prove the NP-hardness when $\alpha(\cdot, \cdot)$ 
	is upper bounded by $k^{1/3}/3$. 
	As discussed at the end of the proof, the reduction can be easily adjusted when the upper bound 
	of $\alpha(\cdot, \cdot)$ is $k^{1-\varepsilon}$, for constant $\varepsilon$.
	
	Let us consider an instance of 3-Terminal Node Cut, i.e., a graph $G(V,E)$, with $|V(G)|=n$, a set 
	$S=\{v_1, v_2, v_3\}$ of non adjacent vertices of $G$, and an integer $q$. We may assume that $q\leqslant n-3$, otherwise the question is trivial. Next we give a construction
	of an appropriate instance for Inequity Aversion Pricing. 
	
	Let $H$ be the graph obtained from $G$ as follows. We replace every vertex $v\in S$ 
	by $n^3$ vertices, where each such vertex has the same neighbors as $v$, i.e., if 
	$u_v$ is a vertex in the bundle of vertices replacing $v$, then for every edge $(v, x)\in E(G)$ we add the edge $(u_v, x)$ to $E(H)$. 
	For any $v\in S$, we call such a set of vertices in $H$ a $v$-bundle.  
	The set of prices is $\{\perp, 1,\allowbreak 2, \ldots, k\}$, where $k= n^3+n^2$.
	Finally, for any $(u, v)\in E(H)$ we set $\alpha(u, v)$ and $\alpha(v, u)$ arbitrarily, as long as they are at most $k^{1/3}/3$.
	Note that $|V(H)|= n - 3 + 3 n^3$, and $|E(H)| \leqslant  |E(G)| + 3(n-1)n^3 \leqslant 3n^{4}$.
	
	Next we define the single value revenue functions for the vertices of $H$. For every $v\in V(G)\setminus S$, let $\val(v)=n^3 + n^2$, and for every $v_i\in S$, let $\val(u_{v_i})=n^3+\frac{i-1}{2}n^2$
	for all $u_{v_i}$ in the $v_i$-bundle. 
	We show below that $G$ has a subset of at most $q$ vertices that separate all the 
	vertices of $S$, if and only if there is a feasible choice of prices for
	the vertices of $H$ that gives revenue at least $R_q$, where
	$R_q = \left( n-3 -q\right) n^3 + \sum_{i=1}^{3} n^3\left( n^3+\frac{i-1}{2}n^2\right)$. % \,.\]
	
	One direction is easy. Let $A$ be a subset of at most $q$ vertices of $G$ that 
	separate the three vertices of $S$. For all $v\in A$  we put a discontinuity on the corresponding $v$ in $H$. If we
	think of these vertices as removed from $H$, this creates several connected components. For any other
	vertex $u\in V(H)$, if $u$ is in the same component as some $v_i$-bundle (or itself is one of the vertices of the $v_i$-bundle), set its price to
	$n^3+\frac{i-1}{2}n^2$, otherwise set its price to $n^3+n^2$. Notice that any vertex without a 
	discontinuity produces revenue at least $n^3$, while any vertex $u_{v_i}$ in a $v_i$-bundle 
	with $v_i\in S$ produces revenue exactly $n^3+\frac{i-1}{2}n^2$. 
	Now, it is straightforward to check that this price vector $\mathbf p$ is feasible 
	and gives enough revenue:
	$R(\mathbf p) =  \sum_{u\in V(H)}  R(u) \geqslant \left( n-3 -q\right) n^3 + \sum_{i=1}^{3} n^3\left( n^3+\frac{i-1}{2}n^2\right) = R_q $.
	
	For the opposite direction we begin with a couple of observations. 
	Assume that there is  a price vector $\mathbf p_*$ that gives revenue at least $R_q$. 
	We claim that $ \mathbf p_*$ can have only a few discontinuities.
	
	\begin{myclaim}\label{claim:np-1}
		There is no feasible price vector $\mathbf p$ with $R(\mathbf p)\geqslant R_q$ 
		and more than $q$ discontinuities.
	\end{myclaim}
	\begin{proofof}{Proof of Claim \ref{claim:np-1}} %\renewcommand{\qed}{\hfill $\triangleleft$}
		Let $\mathbf p$ be a feasible price vector with at least 
		$q+1$ discontinuities. Notice that any vertex without a discontinuity produces revenue 
		at most $n^3 + n^2$ and, in particular, any vertex $u_{v_i}$ in a $v_i$-bundle 
		with $v_i\in S$ produces revenue at most $n^3+\frac{i-1}{2}n^2$. The maximum possible revenue for $\mathbf p$ is
		\begin{IEEEeqnarray*}{rCl}
			R(\mathbf p) & \leqslant & \left( n-3 \right) \left( n^3 + n^2\right)  + \sum_{i=1}^{3} n^3\left( n^3+\frac{i-1}{2}n^2\right) - (q+1) n^3  \\
			& = &  R_q + \left( n-3 \right)n^2 - n^3 < R_q \,, 
		\end{IEEEeqnarray*}    
		thus proving the claim. %\qed
	\end{proofof}

	One immediate implication of Claim \ref{claim:np-1} is that for any $v\in S$ not every vertex in the 
	$v$-bundle has price $\perp$. This holds because the  $v$-bundle has $n^3$ vertices 
	and only $q \leqslant n-3$ of them can get $\perp$. This is crucial, because if we think
	of the vertices with price $\perp$ as removed from $H$, then no two vertices are separated 
	because of discontinuities in the $v$-bundles. In particular, we can completely ignore those
	discontinuities with respect to connectivity.

	Let $D_{\mathbf p} = \{v\in V(G)\setminus S\ |\ \mathbf p_{v} =\, \perp\}$, i.e., $D_{\mathbf p}$ 
	is the set of non terminal vertices in $G$ that their corresponding vertices in $H$ have 
	discontinuities in $\mathbf p$.
	So far, by Claim \ref{claim:np-1}, we have that $|D_{\mathbf p_*}|\leqslant q$. What is left to 
	be shown is that these discontinuities separate the $v$-bundles, for any $v\in S$.
	
	\begin{myclaim}\label{claim:np-2}
		There is no feasible price vector $\mathbf p$ such that 
		$R(\mathbf p)\geqslant R_q$, and  for some $v_i, v_j\in S$ vertices from both the 
		$v_i$-bundle and the $v_j$-bundle are in the same connected component of the graph $H' = H-\{v \in V(H)\ |\ v \text{ is not in a} \allowbreak \text{bundle and } \allowbreak \mathbf p_{v} =\, \perp\}$.
	\end{myclaim}
	\begin{proofof}{} %[Proof of Claim \ref{claim:np-2}] \renewcommand{\qed}{\hfill $\triangleleft$}
		Let $\mathbf p$ be a feasible price vector 
		and assume that there exist $v_i, v_j\in S$ such that vertices from both the $v_i$-bundle 
		and the $v_j$-bundle are in the same component of $H'$. First notice that all the vertices in  the $v_i$-bundle and the $v_j$-bundle 
		are in the same component, since vertices in a bundle share the same neighbors. We are going 
		to upper bound the maximum possible revenue for such a price vector. W.l.o.g., assume $i<j$. 
		If all the vertices in the $v_i$-bundle are assigned prices in
		$\{\perp, n^3+ \frac{i-1}{2} n^2 +1, \ldots, k\}$, then they contribute $0$ to the total revenue.
		On the other hand, if there is some vertex in the $v_i$-bundle with price at most
		$n^3+ \frac{i-1}{2} n^2$, then by the feasibility of $\mathbf p$ we have that
		any vertex in the $v_j$-bundle has its revenue upper bounded by $n^3+\frac{i-1}{2} n^2+ \frac{ k^{1/3}}{3}n$.
		To see the latter, notice that if any two vertices from 
		two distinct bundles are connected by a path, then this path has length at most $n$ 
		(like it would in $G$) and therefore their prices can differ by $\frac{ k^{1/3}}{3}n$ at most. 
		We conclude that the loss, compared to the sum of the maximum revenues per vertex, is 
		lower bounded by either $n^3\left( n^3+\frac{i-1}{2}n^2\right)$ or 
		$n^3\left( \frac{j-i}{2}n^2 - \frac{ k^{1/3}}{3}n\right)$
		and therefore by $n^3\left( \frac{1}{2}n^2 - \frac{ k^{1/3}}{3}n\right)$.
		For $n\geqslant 10$, we have 
		\begin{IEEEeqnarray*}{rCl}
			n^3\left( \frac{n^2}{2} -  \frac{\left( n^3 + n^2\right)^{1/3}n}{3}\right) & \geqslant & n^3\left( \frac{n^2}{2} - \frac{1.1^{1/3}n^2}{3}\right) > 0.15 n^2 n^3 \\
			& > & (n-2) n^3 \geqslant (q+1) n^3 \,
		\end{IEEEeqnarray*}
		and we get $R(\mathbf p)<R_q$ in exactly the same way as in the proof of Claim \ref{claim:np-1}. %\qed
	\end{proofof}

	We conclude that $D_{\mathbf p_*}$ is a set of at most $q$ vertices of $G$ that 
	separate all the vertices of $S$. This completes the proof for the case 
	where $\alpha(\cdot, \cdot)$ is upper bounded by $k^{1/3}/3$.
	
	The above reduction, however, generalizes for $\alpha(\cdot, \cdot)$ upper bounded by $k^{1-\varepsilon}$ for any positive constant $\varepsilon$. Let $c\in \mathbb N$ with
	$c > 4/\varepsilon$. 
	If we multiply by $n^c$ all the relevant quantities, i.e., the size of the bundles, $k$, 
	$R_q$, and $\val(v)$ for all $v\in V(H)$, then
	the reduction is identical up to the last part of the proof of Claim \ref{claim:np-2}. Now, 
	the  loss is lower bounded by $n^{c+3}\left( n^{c+2}/2 - n k^{1-\varepsilon}\right)$ and
	it suffices for this quantity to be at least $(q+1)n^{c+3}$ for things to work out.
	So, we need $ n^{c+2}/2 - n k^{1-\varepsilon} \geqslant n-2$ (since $n-2 \geqslant q+1$),
	and it is only a matter of simple calculations to check that this holds. %\qed
\end{proof}

For the special case where all the  differences are $0$, we  show that the problem is 
APX-hard. 
In doing so, we prove that 3-Terminal Node Cut is MAX SNP-hard, and thus APX-hard. 
As noted already, MAX SNP-hardness of 3-Terminal Node Cut is discussed ---but not explicitly proved--- in \cite{GVY04}. Here, having this reduction is crucial, and we have therefore obtained an explicit construction, since eventually we need to show that 3-Terminal Node Cut restricted in a specific set of instances is MAX SNP-hard (Corollary \ref{cor:3TNC-G}).

\begin{theorem}\label{MNC-hard}
	Multi-Terminal Node Cut is MAX SNP-hard even for $3$ terminals and all the weights equal to $1$.
\end{theorem}
\begin{proof}
	We prove the result for $3$ terminals. The extension to more follows immediately.
	Proofs of MAX SNP-hardness involve linear reductions. Let $A$ and $B$ be two optimization problems. We say that $A$ linearly reduces to $B$ if there are two polynomial time computable functions $f$ and $g$ and constants $c_{\alpha}, c_{\beta} >0$ such that
	\begin{itemize}
		\item Given an instance $a$ of $A$, $f$ produces an instance $b=f (a)$ of $B$ such that $\OPT_B(b) \leqslant c_{\alpha} \OPT_A(a)$, and
		\item Given $a$, $b=f (a)$, and any solution $y$ of $b$, $g$ produces a solution $x$ of $a$ such that $|cost_A(x) - \OPT_A(a)| \leqslant c_{\beta}|cost_B(y) - \OPT_B(b)|$.
	\end{itemize}
	
	The reduction is from the unweighted version of \emph{3-Terminal Cut}: Given a graph $G(V,E)$ and 
	a set $S= \{v_1, v_2, v_3\}\subseteq V$, find a minimum cardinality set of edges that 
	can be deleted, so that any two vertices of $S$ are in different connected components of 
	the resulting graph.
	3-Terminal Cut was shown to be MAX SNP-hard in \cite{DJPSY94} even 
	when all the weights equal to $1$, which is essentially the unweighted version defined above. 
	
	Consider an instance of 3-Terminal Cut, i.e., a graph $G(V,E)$ with $|V(G)|=n$ and 
	a set of non adjacent terminals $S=\{v_1, v_2, v_3\}$. 
	We first describe the function $f$ in the definition of the linear reduction. 
	Let $H$ be the graph obtained from $G$ as follows:
	\begin{enumerate}
		\item Replace each edge $e$ by a path of length two, the middle vertex of which  
		we denote by $v_e$. 
		\item Replace every ``old'' vertex $v$ by a $v$-bundle of $\deg_G(v)+1$ vertices 
		(see also the proof of Theorem \ref{np-comp}), 
		where each such vertex has the same neighbors as $v$ in the graph constructed at step $1$. 
		That is, put an edge between $u_v$ and $v_e$ if $u_v$ is a vertex in the $v$-bundle 
		and $e$ is incident to $v$.
	\end{enumerate}
	Also, let $S'= \{u_1, u_2, u_3\}$, where $u_i$ is an arbitrarily chosen 
	vertex from the $v_i$-bundle. Define $f\left( (G, S)\right)  = \left( H, S'\right) $. Clearly, $f$ is 
	polynomial time computable.
	
	Next we define the function $g$ in the definition of the linear reduction.
	Given a vertex cut $Y$ in $H$ that separates the vertices in $S'$, 
	first we transform it to an appropriate vertex cut $Y'$ that separates the 
	vertices in $S'$ and contains no vertices from any $v$-bundle.
	\begin{enumerate}
		\item While there is a whole $v$-bundle contained in the vertex cut, remove those vertices from the cut and add all of their neighbors instead.
		\item While there is some vertex from a $v$-bundle in the cut, just remove this vertex from the cut.
	\end{enumerate}
	Notice that in one iteration of step $1$ the connectivity is not improved and the size of the 
	vertex cut is reduced. The latter holds because $\deg_G(v)+1$ vertices were removed from the cut and at 
	most $\deg_G(v)$ were added. Similarly, in one iteration of step $2$ the connectivity 
	is not improved and the size of the vertex cut is reduced. Now the latter is obvious, but
	to see that the connectivity is not improved, notice that the removal of vertices in some $v$-bundle 
	has an effect in connectivity only if the whole $v$-bundle is removed. Since in step $2$ there are no
	$v$-bundles completely contained in the vertex cut (this was fixed in step $1$), the vertices removed 
	from the cut were not disconnecting anything to begin with.
	We conclude that $Y'$ is indeed a vertex cut that separates the vertices in $S'$ and moreover 
	$|Y'|\leqslant |Y|$.
	
	Now, that $Y'$ contains only vertices outside the $v$-bundles, i.e., only vertices that correspond to
	edges of $G$, it is straightforward to define an edge cut in $G$ that separates the vertices in $S$. Let 
	$ X=\{e\in E(G)\ |\ v_e\in Y'\}$,
	i.e., $X$ is the set of edges in $G$ that their corresponding vertices in $H$ are in the vertex cut. 
	Define $g\left( (G, S) , \left( H, S'\right), Y\right)$ to be equal to $X$; clearly, $g$ is  polynomial time computable.
	It remains to be shown that $X$ separates the vertices in $S$. Assume not; then there exists 
	some $v_i-v_j$ path $p=(v_i,x_1, x_2, \ldots, x_k, v_j)$ in $G-X$ for $v_i, v_j \in S$, with $i\neq j$.
	This, however, directly transforms to a $u_i-u_j$ path $p'=(u_i, v_{(v_i, x_1)},x_1', v_{(x_1, x_2)}, x_2', \ldots, x_k', v_{(x_k, v_j)}, u_j)$ in $H-Y'$, where $x_{\ell}'$ is an arbitrary vertex in the $x_{\ell}$-bundle. This is a contradiction. Thus, $X$ is a  cut that separates the vertices in $S$.
	
	Next, we prove that $\OPT_{3TNC}(H)\leqslant \OPT_{3TC}(G)$ (to improve readability we drop the subscripts). 
	Notice that any cut $A$ in $G$ that separates the vertices 
	of $S$ gives the vertex cut $B = \{v_e\in V(H)\ |\ e\in A\}$ in $H$ that separates the 
	vertices of $S'$. Since $|B|=|A|$, and by taking $|A|$ to be an optimal cut, we have $\OPT(H)\leqslant \OPT(G)$. This also implies that $c_{\alpha}= 1$ works.
	
	Finally, since $|X|=|Y'|$, we have $|X| - \OPT(G) \leqslant |Y'| - \OPT(H) \leqslant |Y| - \OPT(H)$, i.e.,  $c_{\beta} = 1$ works. We conclude that the unweighted version of 3-Terminal Node Cut is MAX SNP-hard. %\qed
\end{proof}

As proved in \cite{KMSV98}, APX is the closure of MAX SNP under PTAS reductions (introduced by \cite{CT00}). Therefore, any MAX SNP-hard problem is also APX-hard.
Let $\mathcal I$ be the set of instances of 3-Terminal Node Cut that can be the result of 
the composition of the reduction of Theorem \ref{MNC-hard} with the linear reduction 
from Max Cut to 3-Terminal Cut, presented in \cite{DJPSY94}. 
The next corollary follows directly. 

\begin{corollary}\label{cor:3TNC-G}
	3-Terminal Node Cut is MAX SNP-hard, and thus APX-hard, even when restricted on instances in $\mathcal I$. 
\end{corollary}

Corollary \ref{cor:3TNC-G} is a crucial step towards our goal, since instances in 
$\mathcal I$ are guaranteed to have only ``large'' vertex cuts that separate the terminals.

\begin{lemma}\label{lem:large-cut}
	Let $(G, S, q)\in \mathcal I$. Then, any feasible 3-Terminal Node Cut solution for $(G, S, q)$ 
	has size greater than $\frac{1}{14}|V(G)|$. 
\end{lemma}
\begin{proof}
	Let $G_0$ be a graph with $n_0$ vertices and $m_0$ edges. The reduction of \cite{DJPSY94} 
	adds $3$ terminals and, furthermore, for each edge adds $4$ new vertices and $102$ new edges.
	In fact, each edge is replaced with the gadget shown in Figure \ref{fig:3way} (Figure 11 of \cite{DJPSY94}), 
	where $s_1, s_2, s_3$ are identified with the terminals and $x, y$ with the endpoints of the edge.
	Then, each of the $12$ edges with weight $4$ is replaced by $4$ paths of length $2$. The resulting 
	graph $G_1$, has $n_1= n_0 + 3 + 52 m_0$ vertices and $m_1 = 102 m_0$ edges.
	
	Our reduction adds $1$ new vertex for each edge, and then replaces each one of the old
	vertices with $\deg_{G_1}(v)+1$ new vertices. The number of vertices of the resulting graph $G_2$ is 
	$n_2= \sum_{v\in V{(G_1)}}\left( \deg_{G_1}(v)+1 \right) + m_1 = n_1 + 3 m_1= n_0 + 3 + 358 m_0 < 378 m_0$.
	
	By the proof of Theorem $3$ in \cite{DJPSY94}, we have that any cut in $G_1$ that 
	separates the $3$ terminals has size at least $27 m_0$.
	Using $g$ from our reduction, however, we can transform a vertex cut that 
	separates the $3$ terminals in $G_2$ into a cut of the same cardinality
	that separates the $3$ terminals in $G_1$. 
	Thus, any vertex cut that  separates the $3$ terminals in $G_2$ has size at least $27 m_0$.
	To complete the proof, notice that $27 m_0 > 27 n_2/378 = n_2 /14$. %\qed
\end{proof}

\begin{figure}[h]
	\centering
	\includegraphics[scale = 0.7]{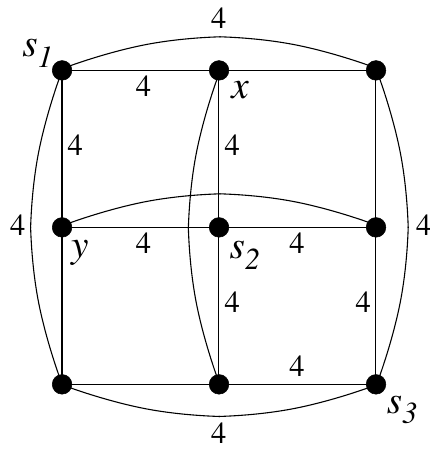}
	\caption{The gadget that ``replaces'' every edge in the linear reduction of from Max Cut to 3-Terminal Cut \cite{DJPSY94}.}\label{fig:3way}
\end{figure}

\begin{theorem}\label{APX-hard}
	Inequity Aversion Pricing for single value revenue functions is APX-hard when $\alpha(e)= 0$ for all $e\in E(G)$. 
\end{theorem}
\begin{proof} %[Proof of Theorem \ref{APX-hard}]
	We use a PTAS reduction to prove the APX-hardness. 
	Let $A$ and $B$ be two NPO problems. Here assume that $A$ is a minimization and $B$ 
	is a maximization problem. We say that $A$ is PTAS-reducible to $B$ if there exist 
	three computable functions $f$, $g$, and  $c$ such that
	\begin{itemize}
		\item For any instance $x$ of $A$ and any $r>1$, $f(x, r)$ is an instance of $B$ computable in time polynomial in $|x|$. 
		\item For any instance $x$ of $A$, any $r>1$, and any feasible solution $y$ of $f(x, r)$, $g(x, y, r)$ is a feasible solution of $A$ computable in time polynomial in both $|x|$ and $|y|$.
		\item $c: (1, \infty) \to (0, 1)$
		\item For any instance $x$ of $A$, any $r>1$, and any feasible solution $y$ of $f(x, r)$,
		\begin{center}
			$cost_B(y)\geqslant c(r)\cdot \OPT_B\left( f(x, r)\right)$ implies $cost_A\left( g(x, y, r)\right) \leqslant r\cdot \OPT_A\left( x\right)$.
		\end{center}
	\end{itemize}
	The reduction is from the restriction of 3-Terminal Node Cut on $\mathcal I$. 
	It is similar to the reduction in the proof of Theorem \ref{np-comp}, but here 
	all the parameters are carefully tuned. 
	Consider an instance of 3-Terminal Node Cut, i.e., a graph $G(V,E)$ with $|V(G)|=n$, 
	a set of non adjacent terminals $S=\{v_1, v_2, v_3\}$, and an integer $q$, such that 
	$(G, S, q)\in \mathcal I$. We describe the function $f$ in the definition of a PTAS reduction. 
	
	For $r>1$, let $\varepsilon = \min\{0.5, r-1\}$ and $t = \left\lceil \frac{42}{\varepsilon} \right\rceil$. 
	Also, let $H = f(G, r)$ be the graph obtained from $G$ by replacing every vertex $v\in S$ 
	by a $v$-bundle of $4t n$ vertices, 
	each such vertex having the same neighbors as $v$. The set of prices is $\{\perp, 1,\allowbreak 2, \ldots, t\}$.
	To define the single value revenue functions, for every $v\in V(G)\setminus S$, 
	let $\val(v)= t$, and for every 
	$v_i\in S$, let $\val(u_{v_i})=t +i - 3$
	for all $u_{v_i}$ in the $v_i$-bundle. We define  $f\left( (G, S, q), r\right)$ 
	to be the above instance. Clearly, $f$ is computable in polynomial time in $n$.
	
	Next we define the function $g$ in the definition of a PTAS reduction.
	Given a feasible price vector $\mathbf p$ for $H$, first we transform it to an appropriate
	feasible price vector $\mathbf p'$.
	\begin{enumerate}
		\item While there is a whole $v$-bundle only with discontinuities, set price 
		$\val(u_{v})$ to all the vertices $u_v$ in this $v$-bundle and $\perp$ to all of their neighbors.
		\item Consider the graph after we remove all the vertices with price $\perp$. 
		While there are $i, j$ (assume $i<j$) such that vertices from both the $v_i$-bundle 
		and the $v_j$-bundle are in the same connected component:
		\begin{itemize}
			\item[--] If all the vertices in the $v_i$-bundle are assigned prices in
			$\{\perp, \val(u_{v_i})+1, \ldots, t\}$, then set price $\val(u_{v_i})$ to all the 
			vertices in the $v_i$-bundle and $\perp$ to all of their neighbors. 
			\item[--] Otherwise, set price $\val( u_{v_j}) $ to all the 
			vertices in the $v_j$-bundle and $\perp$ to all of their neighbors.
		\end{itemize}
		
	\end{enumerate}
	Then, we use this price vector $\mathbf p'$ in order to define a solution $D$ to the 3-Terminal Node Cut
	instance. Let $D = \{v\in V(G)\setminus S\ |\ \mathbf p'_{v} =\, \perp\}$, i.e., 
	$D$ is the set of non terminal
	vertices in $G$ that their corresponding vertices in $H$ have discontinuities. 
	Again, it is straightforward to see that computing $g\left( (G, S, q), \mathbf p, r\right) $ 
	takes polynomial time in $n$.
	
	It remains to determine the function $c$ in the definition of the reduction; for any $r\in (1, \infty)$  let $c(r) = 1- \frac{1}{20 t^2}$. 
	We need to show that 
	\[   R(\mathbf p)\geqslant c(r)\cdot \OPT(H) \implies cost(D) \leqslant r\cdot \OPT(G)\,. \]
	
	\begin{myclaim}\label{claim:apx-1}
		If $R(\mathbf p)\geqslant c(r) \OPT(H)$, then $\mathbf p'=\mathbf p$, 
		i.e., there is no $v$-bundle only with discontinuities, and every $v$-bundle is in a different 
		connected component.
	\end{myclaim}
	
	\begin{proofof}{} %[Proof of Claim \ref{claim:apx-1}]\renewcommand{\qed}{\hfill $\triangleleft$}
		For the first part, assume that there is a $v$-bundle, where every single vertex gets price $\perp$. 
		We get the following upper bound for $R(\mathbf p)$:
		\[R(\mathbf p) \leqslant \left( n-3 \right) t + \sum_{i=2}^{3} 4 t n \left( t +i-3 \right) \leqslant 8t^2 n-3tn -3t <8t^2 n \,.\]
		On the other hand, there is a feasible price vector that sets all the prices equal to $t-2$, and this 
		way we get a lower bound on $\OPT(H)$. 
		\[\OPT(H) \geqslant (3\cdot 4tn+n-3)(t-2)=12t^2n-23tn-2n-3t+6> 12t^2n-28tn \,.\]
		Notice that, since $\varepsilon \leqslant 0.5$, we have $t\geqslant 84$, and therefore $c(r) >0.99$. 
		So, we have $\frac{R(\mathbf p)}{\OPT(H)}<\frac{8t}{12t-28}<0.8 < c(r) $, which is a contradiction.
		
		For the second part, assume that there are two $v$-bundles in the same component. To arrive
		at a contradiction, it suffices to show that there exists a feasible price vector $\mathbf p''$,
		such that $R(\mathbf p)< c(r)R(\mathbf p'')$ and therefore $R(\mathbf p)< c(r)\OPT(H)$. 
		Let $\mathbf p''$ be the price vector 
		obtained after just one iteration of step 2 in the description of $g$. 
		Assuming that we are talking about the $v_i$-bundle and the $v_j$-bundle, 
		with $i<j$,  the gain in revenue 
		is at least $4 t n \big( \val(u_{v_j})-\val(u_{v_i})\big) \geqslant 4tn$ (see also the
		proof of Claim \ref{claim:np-2} in the proof of Theorem \ref{np-comp}). On the other hand, 
		the loss in revenue is upper bounded by 
		$(n-3)\left( t + i-3 \right) \leqslant tn$. 
		So, $R(\mathbf p'') \geqslant R(\mathbf p)+ 3tn$. Suppose 
		$R(\mathbf p)\geqslant c(r)R(\mathbf p'')$. Then it is a matter of simple calculations to see that
		\[R(\mathbf p)\geqslant c(r)(R(\mathbf p)+ 3tn) \implies R(\mathbf p)\geqslant 60 t^3 n - 3 t n > 57t^3 n\,.\]
		An obvious upper bound for $R(\mathbf p)$ however, is to say that each 
		vertex produces revenue at most $t$, i.e.,
		$ R(\mathbf p) \leqslant  (12tn+ n-3)t< 13tn$. Combining the two, we get the contradiction 
		$R(\mathbf p)> 57t^3 n > 13tn > R(\mathbf p)$. We conclude that 
		$R(\mathbf p)< c(r)R(\mathbf p'')$, that leads to 
		the contradiction $R(\mathbf p)< c(r)\OPT(H)$. Hence, in the graph defined by removing the 
		discontinuities of $\mathbf p$ from $H$, every $v$-bundle is in a different connected component.
	\end{proofof}

	\begin{myclaim}\label{claim:apx-2}
		If $R(\mathbf p)\geqslant c(r)  \OPT(H)$, then $\mathbf p$ has less than $(1+\varepsilon)\OPT(G)$ discontinuities.
	\end{myclaim}
	
	\begin{proofof}{} %[Proof of Claim \ref{claim:apx-2}]\renewcommand{\qed}{\hfill $\triangleleft$}
		Let $\mathbf p$ be a feasible price vector with $R(\mathbf p)\geqslant c(r) \OPT(H)$ 
		and assume that $\mathbf p$ has 
		at least $(1+\varepsilon)\OPT(G)$ discontinuities. Also, consider the feasible price vector 
		$\mathbf p^*$ induced by an optimal cut in $G$, i.e., the price vector that sets $\perp$ in every 
		vertex that has a corresponding vertex removed by the cut in $G$ and then uses optimal 
		single price in each ``connected component''. 
		To get a contradiction, we show that $R(\mathbf p)< c(r)R(\mathbf p^*)$ and therefore $R(\mathbf p) < \allowbreak c(r)\OPT(H)$.
		To obtain a lower bound on $R(\mathbf p^*)$, notice that any vertex without a 
		discontinuity produces revenue at least $t-2$, while any vertex $u_{v_i}$ in a $v_i$-bundle 
		produces revenue exactly $t+i-3$. So, 
		\[R(\mathbf p^*) \geqslant \left( n-3 -\OPT(G)\right) (t-2) + \sum_{i=1}^{3} 4 t n \left( t +i-3 \right)\,.\]  
		To get an upper bound for $R(\mathbf  p)$, notice that each vertex without a discontinuity 
		produces revenue at least $t-2$ and at most $t$, while any vertex $u_{v_i}$ in a $v_i$-bundle 
		produces revenue exactly $t+i-3$, i.e,
		\[ R(\mathbf p) \leqslant \left( n-3 \right) t - (1+\varepsilon)\OPT(G) (t-2) + \sum_{i=1}^{3} 4 t n \left( t +i-3 \right) \,.\]
		We consider the difference $R(\mathbf p) - c(r)R(\mathbf p^*)$, and show it is negative. Recall that Lemma \ref{lem:large-cut} implies that $\OPT(G)\geqslant n/14$. 
		\begin{IEEEeqnarray*}{rCl}
			R(\mathbf p) - c(r)R(\mathbf p^*) & \leqslant & \frac{1}{20 t^2}\Big( \left( n-3 -\OPT(G)\right) (t-2) + \sum_{i=1}^{3} 4 t n \left( t +i-3 \right)\Big)\\
			& & +\, 2(n-3) - \varepsilon\, \OPT(G) (t-2)   \\
			& < &  \frac{1}{20 t^2}\left( n t +  12 t^2 n \right) + 2n - \varepsilon\,\frac{1}{14} n \left( \frac{42}{\varepsilon}-2\right) \\
			& < &  \frac{13n}{20} + 2n - 2.9 n < -0.25 n <0  \,,
		\end{IEEEeqnarray*}
		which leads to contradiction. Thus, $\mathbf p$ has less than $(1+\varepsilon)\OPT(G)$ discontinuities.  
	\end{proofof}
	
	By combining Claim \ref{claim:apx-1}, Claim \ref{claim:apx-2}, and the fact that $1+ \varepsilon \leqslant r$, we directly get that a $c(r)$-approximate solution for $H$ gives an $r$-approximate solution for $G$, thus concluding the proof. %\qed
\end{proof}

\begin{remark}
	The maximum price $k$ in the instance constructed in the 
	proof of Theorem \ref{APX-hard} does not depend on the size of the problem. Given that 
	there is some constant $\rho$ beyond which it is hard to approximate $3$-Terminal Node 
	Cut, this means that there exists some constant $k^*$ for which Inequity Aversion Pricing does not 
	have a PTAS. Note that for such a $k^*$ we do have a constant factor approximation, 
	with factor $H_{k^*}^{-1}$.
\end{remark}

\subsection{Hardness when $k=3$}
We close this section by showing that Inequity Aversion Pricing remains hard even when we only have three possible prices and $\alpha(e) = 0$ for all edges. This identifies the transition from polynomial time solvability, which we have when $k=2$, to NP-hardness as soon as we have a higher number of available prices.

\begin{theorem}\label{NP-hard-3}
	Inequity Aversion Pricing for single value revenue functions is NP-complete when $\alpha(e)= 0$ for all $e\in E(G)$, even if the price set is $P = \{1, 2, 3\}$. 
\end{theorem}

The theorem follows from the fact that the problem is trivially in NP and the next three lemmas, each consisting of a simple reduction. We begin with the definition of two intermediate problems used in those reductions.

\begin{definition}
	The \emph{Tripartite Independent Set} problem is the restriction of {Independent Set} on tripartite graphs. In particular, given a tripartite graph, a tripartition of its vertices, and an integer $q$, is there an independent set of size at least $q$?
\end{definition}

The next problem is a stricter version of our problem, regarding the price that we are allowed to offer to each node.

\begin{definition}
	The \emph{Strict Inequity Aversion Pricing} problem is a variant of 
	Inequity Aversion Pricing in which $\alpha(e)= 0$ for all $e\in E(G)$ and the seller is disallowed to offer a customer a price different from the customer's valuation, i.e.\ $p_v \in \{ \val(v), \bot \}$ for each node $v$.
\end{definition}

\begin{lemma}
	Tripartite Independent Set is NP-hard.
\end{lemma}

\begin{proof}
	We reduce the general Independent Set problem to Tripartite Independent Set using a construction from
	\cite{poljak1974note}.
	
	Given a graph $G$ with $n$ vertices and $m$ edges, we 2-subdivide its edges, i.e., replace each edge with a path
	of length 3, to obtain a graph $H$, which is clearly tripartite. We  call
	the vertices added by 2-subdivisions \emph{new} as opposed to the \emph{old} vertices coming from  $G$. Now $G$ has an independent set of size $q$ if and only if $H$ has an independent set of size $q+m$: Starting from an independent set of
	$G$, we can add to it one of the two new vertices on each 2-subdivided edge.
	Conversely, every independent set of $H$ can be transformed into one that is not
	smaller and contains precisely $m$ new vertices (one for each 2-subdivided
	edge); the old vertices of this independent set then form an independent set of $G$. %\qed
\end{proof}

\begin{lemma}
	Strict Inequity Aversion Pricing with price set $P = \{1, 2, 3 \}$ is NP-hard.
\end{lemma}

\begin{proof}
	We give a reduction from Tripartite Independent Set.
	Given a tripartite graph $G$, a tripartition $V_1, V_2, V_3$ of its vertices, and an integer $q$, we construct an instance of Strict Inequity Aversion Pricing as follows:
	For each vertex $v \in V_i$, we  have a bundle of $6/i$ nodes $v'$ with
	$\val(v') = i$, for $i=1, 2, 3$. For each edge $(u, v)$ of $G$, we add constraints between all pairs
	$(u', v')$ of nodes associated with $u$ and $v$ respectively, setting $\alpha(u', v') =
	\alpha(v', u') = 0$. Call $H$ the resulting graph and let $R_q=6q$.
	
	We claim that $G$ has an independent set of size $q$ if and only if there is a feasible price vector for the above instance that guarantees revenue $R_q$.
	One direction is straightforward. For every vertex $v$ in an independent set of size $q$ in $G$, we set price $\val(v')$ to every vertex $v'$ of the corresponding bundle in $H$. This way no constraint is violated, since we started with an independent set, and each bundle contributes to the total revenue either a value of $6$, if it corresponds to a vertex in the independent set, or $0$, for a total of $6q$.
	
	Conversely, suppose that there is a feasible price vector $\vec{p}$ for $H$ that guarantees revenue $R_q$, for the strict version of Inequity Aversion Pricing.
	Because all nodes $v'$ in a bundle of $H$ have the same neighborhood, they may be given the same offer (i.e., $\val(v')$ or $\bot$). If this is not already the case for $\vec{p}$, we can find such a feasible price vector $\vec{p}'$ by using in each bundle the maximum price that $\vec{p}$ uses on any vertex of this bundle. Since the new prices only go up (without affecting feasibility), $\vec{p}'$ guarantees revenue $R' \ge R_q$.  Moreover, under $\vec{p}'$, a bundle of nodes contributes to the total revenue either $6$ or $0$, regardless of which part $V_i$ their associated vertex
	$v$ belongs to. Let us denote by $S$ the set of vertices of $G$ such that their
	associated nodes \emph{were not} assigned $\bot$; from the construction it
	follows that $S$ is an independent set in $G$, since we have that $\alpha(e) = 0$ for every edge $e$, and therefore, for an edge $(u, v)$ in $G$, it cannot be the case that the associated nodes for both $u$ and $v$ in $H$ were not assigned $\bot$. The size of $S$ is $R'/6 \ge R_q /6 = q$, which concludes the proof. %\qed
\end{proof}

\begin{lemma}
	Inequity Aversion Pricing with
	price set  $P = \{1, 2, 3 \}$ is NP-hard.
\end{lemma}

%\color{blue}

\begin{proof}
	We give a reduction from Strict Inequity Aversion Pricing with 
	price set $\{1, 2, 3 \}$. Consider an instance of the decision version of the problem, i.e.,
	a graph $G$ with edge constraints ($\alpha(\cdot,\cdot) = 0$), and a single-value revenue function for each node, along with a positive integer $t$. Let $n_i = |\{v\in V(G)\ |\ \val(v)=i\}|$ for $i\in\{1, 2, 3\}$. We construct an instance of Inequity Aversion Pricing as follows:
	For each node $v\in G$, we add $\val(v)$ new nodes $v'$ with $\val(v') = \val(v)$ and
	impose constraints $\alpha(v, v') = \alpha(v', v) = 0$, forming a star with $v$
	at its center. Call $H$ the resulting graph and let $t' = t + n_1 + 4n_2 + 9n_3$. 
	
	We first observe that if there is a feasible price vector $\vec{p}$ for $G$ that produces revenue $t$, then we can use it to set the price on the old nodes of $H$, while for each new node $v'$ we set
	its price to $\val(v')$, and the resulting price vector $\vec{p}'$ is still feasible. The feasibility of $\vec{p}'$ follows from the fact that for every old node $v$, $\vec{p}$ sets a price of $\val(v)$ or $\bot$. By the construction of $H$, it is straightforward to see that this way we extract revenue $t'$.
	
	Conversely, suppose that there is a feasible price vector $\vec{p}'$ for $H$ that produces revenue $t' \geq t + n_1 + 4n_2 + 9n_3$. We will construct a feasible price vector $\vec{p}$ for $G$ as follows. For each $v \in G$, $\vec{p}_v = \bot$ if $\vec{p}'_v \neq \val(v)$, and $\vec{p}_v = \vec{p}'_v$ otherwise. Feasibility follows from the feasibility of $\vec{p}'$ for $H$ (we have only introduced more $\bot$s). We next show that $\vec{p}$ gives revenue at least $t$.
	
	To do so, we construct a feasible price vector $\vec{p}''$ for $H$ that produces revenue $t'' \geq t'$ and agrees with $\vec{p}$ on all old nodes. For each old node $v \in H$ such that $\vec{p}'_v \neq \val(v)$, we set $\vec{p}''_v = \bot$, and for all new nodes $v'$ that are in a star with $v$ we set $\vec{p}''_{v'} = \val(v') = \val(v)$. This way we increase the revenue by at least $1$ without sacrificing feasibility. Now the revenue extracted using $\vec{p}''$ on $H$ is at least $t'$ but a total of at most $n_1 + 4n_2 + 9n_3$ is due to new nodes. That is, the revenue extracted from old nodes in $H$ using $\vec{p}''$ is at least $t'-(n_1 + 4n_2 + 9n_3)=t$. Since $\vec{p}''$ agrees with $\vec{p}$ on all old nodes, the revenue extracted using  $\vec{p}$ on $G$ is is at least $t$ as well. %\qed
\end{proof}

%==================================================
%==================================================
%==================================================
\section{A Generalization to Multi-Demand Users}
\label{sec:slope-svrf}
%==================================================
%==================================================
%==================================================

So far, we have always assumed that each node has demand for only one copy of the product. 
A natural generalization is to consider multi-demand users who are interested in receiving a certain number of copies if the price is affordable. For example, someone might want to buy either a certain number of licenses of a video game (because she wants to play the game with her friends) or no license at all.  
This would correspond to a type of inelastic multi-unit demand in the terminology of auctions.
Assume again that there is enough supply of copies to satisfy all the demand, if necessary.
Then, there is a natural way to generalize single value revenue functions to 
capture such simple scenarios.

A revenue function $R_v(\cdot)$ is called a \textit{multi-demand single value revenue function} if there 
exist an integer $s_v$ (the number of copies demanded) and  a value $\val(v)$ such that:
\[R_v(p_v) =
\begin{cases}
s_v p_v & \text{if } \val(v)\geqslant p_v \\
0 & \text{if } \val(v) < p_v
\end{cases}.\]
The intuition here is the same as for the single value revenue functions.

The objective now is again the same. Given a multi-demand single value revenue function for each node, find a feasible price vector $\vec{p}$ that maximizes the total revenue.
We call this problem \emph{Multi-Demand Inequity Aversion Pricing}. As this is a generalization of 
Inequity Aversion Pricing, it is immediate that any negative result for the latter
yields the same negative result for the former. In particular, by Theorems \ref{np-comp} and \ref{APX-hard}, Multi-Demand Inequity Aversion Pricing is NP-hard and APX-hard for the
corresponding edge constraints.

Quite surprisingly, we also prove that when for each user the number of demanded copies is polynomially bounded, there is a strict reduction from 
Multi-Demand Inequity Aversion Pricing to Inequity Aversion Pricing. This directly implies
that any approximation factor achieved for the latter is also achieved for the former.
Therefore, we establish that the two problems are equivalent in terms of approximability.
Note that the theorem holds for general edge constraints. 

\begin{theorem}\label{theorem:strict-reduction}
	Let $q$ be any polynomial. There exists a strict reduction from Multi-Demand Inequity Aversion Pricing with demands bounded by $q(n)$ to Inequity Aversion Pricing.
\end{theorem}

\begin{proof}
	Suppose we are given an instance $I$ of Multi-Demand Inequity Aversion Pricing, i.e., a graph $G(V, E)$, an edge restriction function $\alpha(\cdot, \cdot)$, and for each node $v$ her valuation $\val(v)$ and her demand $s_v$. We are going to construct an equivalent instance $I'$ of Inequity Aversion Pricing.
	The reduction creates $s_v$ copies of $v$ for each $v \in V$ and connects them to each other to create a clique $K_{s_v}$. Edges inside such a clique have $\alpha=0$. For every edge $(u, v)\in E$ all the edges between the vertices of the $u$-clique and the $v$-clique are added with the same restrictions as the original edge. Let $G'=(V', E')$ be the resulting graph. If $s_{\max}= \max_{v\in V} s_v$ then we have 
	$|V'|\leqslant n s_{\max}$ and $|E'|\leqslant (n+m) s_{\max}^2$.
	
	We use $\OPT'$ and $\OPT$ to denote the optimal revenue of this instance and of the original, respectively.
	Our goal is to show that for any price vector $\mathbf p'$ for $I'$ we can efficiently find a feasible price vector $\mathbf p$ for $I$ with such that $\frac{R(\mathbf p)}{\OPT} \geqslant  \frac{R'(\mathbf p')}{\OPT'}$. 
	We begin by proving that $\OPT = \OPT'$.
	
	\begin{myclaim}\label{claim:equal-prices-inside-Kv}
		An optimal price vector $\mathbf p'$ for $I'$  sets the same price for all vertices inside each $v$-clique.
	\end{myclaim}
	\begin{proofof}{} %[Proof of Claim \ref{claim:equal-prices-inside-Kv}]\renewcommand{\qed}{\hfill $\triangleleft$}
		Note that $\alpha=0$ inside each $v$-clique, so all these vertices have the same common price $p'_v$ or $\bot$. If there are $x,y$ in a  $v$-clique such that $p'_x \neq \bot \wedge p'_y = \bot$ then by setting $p'_y = p'_x$ we obtain a new feasible price vector for $I'$ that gives greater revenue than $\mathbf p'$, which contradicts its optimality.
	\end{proofof}
	
	By Claim \ref{claim:equal-prices-inside-Kv}, we directly obtain a feasible solution for $I$ with   revenue $\OPT$ by setting $p_v$ equal to the common price from the $v$-clique. Therefore, $\OPT \geqslant \OPT'$.
	
	On the other hand, each feasible price vector $\mathbf p$ for $I$ can be adopted as a feasible price vector $\mathbf p'$ for $I'$ with the same revenue. To see that, just set the same price $p'_{u_v} = p_v$ for each copy $u_v$ of $v$ in the $v$-clique of $G'$. All edge constraints are satisfied, so the solution is feasible, and it clearly gives the same revenue.
	By taking $\mathbf p$ to be an optimal price vector for $I$, the above implies that $\OPT' \geqslant \OPT$.
	We conclude that $\OPT = \OPT'$.
	
	Finally, we need the following.
	\begin{myclaim}\label{claim:P-geq-Psimple}
		Each feasible price vector $\mathbf p'$ for $I'$ can be transformed into a feasible price vector $\mathbf p$ for $I$ with at least the same revenue.
	\end{myclaim}
	\begin{proofof}{} %[Proof of Claim \ref{claim:P-geq-Psimple}]\renewcommand{\qed}{\hfill $\triangleleft$}
		For each $u \in V$, if $V_u$ is the set of vertices in the $u$-clique of $G'$, define $u^* = \argmax_{x \in V_u} p'_x$. Then, set $p_u = p'_{u^*}$. Such a $\mathbf p$ is feasible for $I$ because $\forall (u,v) \in E,\ \ \alpha(u,v) = \alpha(u^*,v^*)$, where $v^*$ is any vertex in $V_v$, and the constraint $\alpha(u^*,v^*)$ is already satisfied by $\mathbf p'$. It is straightforward that $R(\mathbf p) \geqslant R'(\mathbf p')$.
	\end{proofof}
	
	\noindent For the price vector described in the proof of Claim \ref{claim:P-geq-Psimple}, we have
	\[\frac{R(\mathbf p)}{\OPT} \geqslant \frac{R'(\mathbf p')}{\OPT} = \frac{R'(\mathbf p')}{\OPT'}\,, \]
	which completes the proof. %\qed
\end{proof}

It would be interesting to determine whether the hardness of the problem changes when the demands are not
polynomially bounded, although such functions are not very realistic in our setting.
Notice, however, that even in that case it is not hard to obtain a $\frac{1}{H_k}$-approximation 
in polynomial time by using the best single-price solution. In fact, we still have a $\frac{1}{H_r}$-approximation, where $r = \min\{n,v_{max}\}$.

\section*{Concluding remarks}
We studied a revenue maximization problem under inequity aversion for the natural class of single-value revenue functions. Apart from establishing the first hardness results for this class, we also derived approximation algorithms based on combinatorial and graph-theoretic tools, which improve the state of the art when the set of available prices is small. We find this to be a realistic setting as special price offers are usually small in number, derived by specific discount and promotion policies.  
Clearly, the most interesting open problem is to resolve the approximability for general $k$, i.e., can we have a better than $O(1/H_k)$-approximation for large $k$? Exploring further models of negative externalities is another attractive direction that has not been given as much attention as the case of positive externalities.

\section*{Acknowledgements} 
This research was supported by National Science Centre, Poland, 2015/17/N/ST6/03684. It was also supported by the European Union (European Social Fund - ESF) and Greek national funds through the Operational Program ``Education and Lifelong Learning'' of the National Strategic Reference Framework (NSRF) - Research Funding Program: THALES.

%==================================================
%==================================================
%==================================================
%% BibTeX users please use one of
%%\bibliographystyle{spbasic}      % basic style, author-year citations
%%\bibliographystyle{spphys}       % APS-like style for physics
\bibliography{pricing2}   % name your BibTeX data base
%==================================================
%==================================================
%==================================================

\end{document}